\documentclass[3p,10pt,a4paper,twoside,fleqn,sort&compress]{filomat}
\usepackage{amsmath,latexsym}
\usepackage{amsthm}
\usepackage[varg]{pxfonts}
\usepackage{verbatim}

\newtheorem{theo}{Theorem}[section]
\newtheorem{corollary}{Corollary}[section]
\newtheorem{lema}{Lemma}[section]

\newcommand{\FF}{\mathcal{F}}
\newcommand{\GG}{\mathcal{G}}
\newcommand{\HH}{\mathcal{H}}

\newcommand{\pmu}{\partial_{\mu}}
\newcommand{\pnu}{\partial_{\nu}}
\newcommand{\nmu}{\nabla_{\mu}}
\newcommand{\nnu}{\nabla_{\nu}}

\newcommand{\dx}{\mathrm{d}^nx}
\newcommand{\Real}{\mathbb{R}}
\newcommand{\FilVol}{xx}
\newcommand{\FilYear}{yyyy}
\newcommand{\FilPages}{zzz--zzz}
\newcommand{\FilDOI}{(will be added later)}
\newcommand{\FilReceived}{Received: dd Month yyyy; Accepted: dd Month yyyy}

\begin{document}

\title{Some Cosmological Solutions of a Nonlocal Modified Gravity}

\author[affil1]{Ivan Dimitrijevic}
\ead{ivand@matf.bg.ac.rs}
\author[affil2]{Branko Dragovich}
\ead{dragovich@ipb.ac.rs}
\author[affil3]{Jelena Grujic}
\ead{jelenagg@gmail.com}
\author[affil1]{Zoran Rakic}
\ead{zrakic@matf.bg.ac.rs}
\address[affil1]{Faculty of Mathematics, University of Belgrade, Studentski trg 16,  Belgrade, Serbia}
\address[affil2]{Institute of Physics, University of Belgrade, Pregrevica 118, 11080 Belgrade, Serbia}
\address[affil3]{Teacher Education Faculty, University of Belgrade, Kraljice Natalije 43, Belgrade, Serbia}
\newcommand{\AuthorNames}{Ivan Dimitrijevic et al.}

\newcommand{\FilMSC}{83D05, 53B21,  53B50; Secondary 53C25, 83F05}
\newcommand{\FilKeywords}{nonlocal gravity, cosmological solutions, equations of motion, calculus of variations, pseudo-Riemannian manifold. }
\newcommand{\FilCommunicated}{(name of the Editor, mandatory)}
\newcommand{\FilSupport}{Work on this paper was supported by Ministry of Education, Science and Technological Development of the Republic of Serbia, grant No 174012.}

\begin{abstract}
We consider nonlocal modification of the Einstein theory of  gravity in framework of the pseudo-Riemannian
geometry. The nonlocal term has the form $\mathcal{H}(R) \mathcal{F}(\Box)\mathcal {G}(R)$, where
$\HH$ and $\GG$ are differentiable functions of the scalar curvature $R,$ and $ \mathcal{F}(\Box)= \displaystyle \sum_{n =0}^{\infty}
f_{n}\Box^{n}$ is an analytic function of the d'Alambert operator $\Box .$ Using calculus of variations of the action functional,  we derived the corresponding equations of motion. The variation of action is induced by variation of the gravitational field, which is the metric tensor $g_{\mu\nu}$.
Cosmological solutions are found for the case when the Ricci scalar $R$ is  constant.
\end{abstract}

\maketitle

\makeatletter
\renewcommand\@makefnmark%
{\mbox{\textsuperscript{\normalfont\@thefnmark)}}}
\makeatother

\section{Introduction}

Although very successful, Einstein theory of gravity is not a final theory. There are  many its modifications,
which are motivated by quantum gravity, string theory, astrophysics and cosmology (for a  review, see \cite{clifton}).
One of very promising directions of research is {\it nonlocal modified gravity} and its applications to cosmology (as a review, see
 \cite{nojiri,woodard} and \cite{dragovich0}). To solve cosmological Big Bang singularity, nonlocal gravity with replacement $R \to
 R + C R \mathcal{F}(\Box)R$  in the Einstein-Hilbert action was proposed in \cite{biswas1}. This nonlocal model is further elaborated
 in the series of papers \cite{biswas2,biswas3,biswas3+,koshelev,dragovich1,dragovich2,biswas4}.

In  \cite{dragovich3} we introduced a new approach to nonlocal gravity given by the  action
\begin{equation} \label{eq-1.2}
S =  \int d^{4}x \sqrt{-g}\Big(\frac{R}{16 \pi G} +  R^{-1}
\mathcal{F}(\Box) R  \Big),
\end{equation}
 where  the d'Alembert operator is $\Box = \frac{1}{\sqrt{-g}} \partial_{\mu}
\sqrt{-g} g^{\mu\nu} \partial_{\nu},$ $\, g = det(g_{\mu\nu}) .$  The  nonlocal term $R^{-1}
\mathcal{F}(\Box) R = f_0 + R^{-1} \sum_{n =1}^{\infty} f_{n}\Box^{n} R$ contains $f_0$ which can be connected with the cosmological constant as $f_0 = - \frac{\Lambda}{8 \pi G}.$ This term is also invariant under transformation $R \to C R ,$ where $C$ is a constant, i.e. this nonlocality does not depend on  magnitude of the scalar curvature $R \neq 0 .$

In this paper we consider $n$-dimensional pseudo-Riemannian manifold $M$ with metric $g_{\mu\nu}$ of signature $(n_-,n_+)$.  Our nonlocal gravity model here is
larger than \eqref{eq-1.2}  and given by the action
\begin{equation} \label{lag:1}
S = \displaystyle \int_M \Big(\frac{R-2 \Lambda}{16\pi G} + \mathcal{H}(R) \mathcal{F}(\Box)\mathcal {G}(R) \Big)\sqrt{|g|} \dx,
\end{equation}
which is a functional of metric (gravitational field) $g_{\mu\nu},$ where   $\HH$ and $\GG$ are differentiable functions of
the scalar curvature $R$, and $\Lambda$ is cosmological constant.

\section{Variation of the action functional}

 Let us  introduce the following auxiliary  functionals
\begin{align}
S_0 &= \int_M (R-2\Lambda) \sqrt{|g|} \; \dx, \qquad
S_1 = \int_M \mathcal{H}(R) \mathcal{F}(\Box)\mathcal {G}(R) \sqrt{|g|} \; \dx.
\end{align}
Then the variations of $S_0$ and $S_1$ can be considered separately and the variation of  (\ref{lag:1}) can be expressed as
\begin{equation}
\delta S = \frac{1}{16\pi G} \delta S_0 +  \delta S_1. \label{varS}
\end{equation}

 Note that variations of the metric tensor elements and their first derivatives are zero on the boundary of manifold $M$, i.e. $\delta g_{\mu\nu}|_{\partial M} = 0$, $\delta \partial_\lambda g_{\mu\nu}|_{\partial M} = 0$.

\begin{lema} \label{lema:1}
Let $M$ be a pseudo-Riemannian manifold. Then the following basic relations hold:

\begin{align*}
\frac{\partial g^{\mu\nu}}{\partial x^{\sigma}}&= -g^{\mu\alpha}\Gamma_{\sigma\alpha}^{\nu}- g^{\nu\alpha}\Gamma_{\sigma\alpha}^{\mu},
& \delta g &= g g^{\mu\nu}\delta g_{\mu\nu}= - g g_{\mu\nu}\delta g^{\mu\nu}, \\
\Gamma_{\mu\nu}^{\mu}&= \frac{\partial}{\partial x^{\nu}}\ln \sqrt{|g|},
& \delta \sqrt{|g|} &= -\frac 12  g_{\mu\nu} \sqrt{|g|}\delta g^{\mu\nu}, \\
\Box &= \nabla^\mu \nmu = \frac 1{\sqrt{|g|}} \pmu(\sqrt{|g|} g^{\mu\nu}\pnu),
& \delta R &= R_{\mu\nu} \delta g^{\mu\nu} + g_{\mu\nu} \Box \delta g^{\mu\nu} - \nmu\nnu \delta g^{\mu\nu}.
\end{align*}
\end{lema}


\begin{lema} \label{lema:2}
On the manifold $M$ holds \; $\displaystyle\int_M g^{\mu \nu} \delta R_{\mu\nu}\sqrt{|g|} \;\dx = 0$.
\end{lema}

\begin{proof}

Let $W^{\nu}= -g^{\mu\alpha}\delta\Gamma_{\mu\alpha}^{\nu}+
g^{\mu\nu}\delta\Gamma_{\mu\alpha}^{\alpha}$. Then it follows
\begin{equation}
\frac{1}{\sqrt{|g|}}\frac{\partial}{\partial
x^{\nu}}(\sqrt{|g|}W^{\nu})= \frac{\partial W^{\nu}}{\partial
x^{\nu}}+ W^{\nu}\frac{1}{\sqrt{|g|}}\frac{\partial
\sqrt{|g|}}{\partial x^{\nu}}.
\end{equation}
Using  Lemma \ref{lema:1} we get
\begin{align}
\frac{1}{\sqrt{|g|}}\frac{\partial}{\partial
x^{\nu}}(\sqrt{|g|}W^{\nu})&= - \frac{\partial}{\partial
x^{\nu}}(g^{\mu\alpha}\delta \Gamma_{\mu\alpha}^{\nu})+
\frac{\partial}{\partial x^{\nu}}(g^{\mu\nu}\delta
\Gamma_{\mu\alpha}^{\alpha}) +( -g^{\mu\alpha}\delta\Gamma_{\mu\alpha}^{\nu}+ g^{\mu\nu}\delta\Gamma_{\mu\alpha}^{\alpha})\Gamma_{\nu\beta}^{\beta} \nonumber \\
&= - \frac{\partial g^{\mu\alpha}}{\partial
x^{\nu}}\delta\Gamma_{\mu\alpha}^{\nu}- g^{\mu\alpha}\delta
\frac{\partial \Gamma_{\mu\alpha}^{\nu} }{\partial x^{\nu}}+
\frac{\partial g^{\mu\nu}}{\partial
x^{\nu}}\delta\Gamma_{\mu\alpha}^{\alpha}+ g^{\mu\nu}\delta
\frac{\partial \Gamma_{\mu\alpha}^{\alpha}}{\partial x^{\nu}} +( -g^{\mu\alpha}\delta\Gamma_{\mu\alpha}^{\nu}+
g^{\mu\nu}\delta\Gamma_{\mu\alpha}^{\alpha})\Gamma_{\nu\beta}^{\beta}.
\end{align}
Moreover,  using again Lemma \ref{lema:1} we obtain
\begin{align}
\frac{1}{\sqrt{|g|}}\frac{\partial}{\partial
x^{\nu}}(\sqrt{|g|}W^{\nu})&=g^{\alpha\beta}\Gamma_{\nu\beta}^{\mu}\delta\Gamma_{\mu\alpha}^{\nu}+ g^{\mu\beta}\Gamma_{\nu\beta}^{\alpha}\delta\Gamma_{\mu\alpha}^{\nu} -g^{\beta\nu}\Gamma_{\nu\beta}^{\mu}\delta\Gamma_{\mu\alpha}^{\alpha}-g^{\mu\beta}\Gamma_{\nu\beta}^{\nu}\delta\Gamma_{\mu\alpha}^{\alpha}
-g^{\mu\alpha}\delta \frac{\partial
\Gamma_{\mu\alpha}^{\nu}}{\partial x^{\nu}}+ g^{\mu\nu}\delta
\frac{\partial \Gamma_{\mu\alpha}^{\alpha}}{\partial x^{\nu}}
\nonumber \\
&-g^{\mu\alpha}\Gamma_{\nu\beta}^{\beta}\delta\Gamma_{\mu\alpha}^{\nu}+
g^{\mu\nu}\Gamma_{\nu\beta}^{\beta}\delta\Gamma_{\mu\alpha}^{\alpha}
\nonumber \\
&= g^{\mu\nu}\Big(-\delta \frac{\partial
\Gamma_{\mu\nu}^{\alpha}}{\partial x^{\alpha}} +  \delta
\frac{\partial \Gamma_{\mu\alpha}^{\alpha}}{\partial x^{\nu}}+
\Gamma_{\alpha\mu}^{\beta}\delta\Gamma_{\beta\nu}^{\alpha} +
\Gamma_{\nu\beta}^{\alpha}\delta\Gamma_{\mu\alpha}^{\beta} -\Gamma_{\nu\mu}^{\beta}\delta\Gamma_{\beta\alpha}^{\alpha} -
\Gamma_{\beta\alpha}^{\alpha}\delta\Gamma_{\mu\nu}^{\beta}\Big)
\nonumber \\
&= g^{\mu\nu}\delta \Big(- \frac{\partial
\Gamma_{\mu\nu}^{\alpha}}{\partial x^{\alpha}} + \frac{\partial
\Gamma_{\mu\alpha}^{\alpha}}{\partial x^{\nu}}+
\Gamma_{\alpha\mu}^{\beta}\Gamma_{\beta\nu}^{\alpha} -
\Gamma_{\beta\alpha}^{\alpha}\Gamma_{\mu\nu}^{\beta}\Big)=
g^{\mu\nu}\delta R_{\mu\nu}.
\end{align}
Finally, we have
\begin{align}
g^{\mu\nu}\delta R_{\mu\nu}
&=\frac{1}{\sqrt{|g|}}\frac{\partial}{\partial
x^{\nu}}(\sqrt{|g|}W^{\nu})\;\mbox{ and }
\int_M g^{\mu\nu}\delta R_{\mu\nu} \sqrt{|g|}\dx = \int_M
\frac{\partial}{\partial x^{\nu}}(\sqrt{|g|}W^{\nu})\dx.
\end{align}
Using the Gauss-Stokes theorem one obtains
\begin{equation}
\int_M \frac{\partial}{\partial x^{\nu}}(\sqrt{|g|}W^{\nu})\dx =
\int_{\partial M}W^{\nu}d\sigma_{\nu}.
\end{equation}
Since $ \delta g_{\mu\nu}=0 $ and $ \delta  \partial_\lambda
g_{\mu\nu}=0 $ at the boundary $\partial M$ we have $W^{\nu}|_{\partial M} = 0$. Then we have
$\int_{\partial M}W^{\nu}d\sigma_{\nu}=0,$ that completes the proof.
\end{proof}

\begin{lema} \label{lema:3}
The variation of $S_0$ is
\begin{equation}\label{varS0}
\delta S_0 =\int_M G_{\mu\nu} \sqrt{|g|} \delta g^{\mu \nu} \; \dx +\Lambda \int_M g_{\mu \nu} \sqrt{|g|} \delta g^{\mu \nu} \; \dx,
\end{equation}
where $ G_{\mu\nu}= R_{\mu\nu}- \frac{1}{2}R g_{\mu\nu} $ is the Einstein tensor.
\end{lema}

\begin{proof}
The variation of $S_0$ can be found as follows
\begin{align}
\delta S_0 &= \int_M \delta ((R-2\Lambda) \sqrt{|g|}) \; \dx
=\int_M \delta (R \sqrt{|g|}) \; \dx - 2\Lambda \int_M \delta \sqrt{|g|} \; \dx \nonumber\\
&=\int_M \left(\sqrt{|g|} \delta R + R \delta \sqrt{|g|} +\Lambda g_{\mu \nu} \sqrt{|g|} \delta g^{\mu \nu} \right) \; \dx \nonumber\\
&=\int_M\left(\sqrt{|g|} \delta (g^{\mu\nu} R_{\mu \nu}) - \frac{1}{2} R \sqrt{|g|} g_{\mu \nu} \delta g^{\mu \nu} + \Lambda g_{\mu \nu} \sqrt{|g|} \delta g^{\mu \nu}\right) \; \dx  \nonumber\\
&=\int_M G_{\mu\nu} \sqrt{|g|} \delta g^{\mu \nu} \; \dx + \Lambda \int_M g_{\mu \nu} \sqrt{|g|} \delta
g^{\mu \nu} \; \dx +\int_M g^{\mu \nu} \delta R_{\mu\nu}\sqrt{|g|} \;\dx.
\label{v-S0}
\end{align}
Using Lemma \ref{lema:2}, from the last equation we obtain the variation of $S_0$.
\end{proof}

\begin{lema}
For any scalar function $h$ we have
\begin{align}
\int_M h \delta R \sqrt{|g|} \; \dx = \int_M \left(h R_{\mu\nu} +
g_{\mu\nu} \Box h -  \nmu \nnu h \right) \delta g^{\mu\nu}
\sqrt{|g|} \; \dx. \label{h:delta R}
\end{align}
\end{lema}

\begin{proof}
Using Lemma \ref{lema:1}, for any scalar function
$h$ we have
\begin{align} \label{13.03:0}
&\int_M h \delta R \sqrt{|g|} \; \dx  = \int_M \left( h R_{\mu\nu} \delta g^{\mu\nu} + h g_{\mu\nu} \Box
\delta g^{\mu\nu} - h \nmu \nnu \delta g^{\mu\nu} \right)
\sqrt{|g|} \; \dx.
\end{align}
The second and third term in this formula can be transformed in
the following way:
\begin{align}
\int_M h g_{\mu\nu} (\Box \delta g^{\mu\nu}) \sqrt{|g|} \; \dx &= \int_M g_{\mu\nu} (\Box h)  \delta g^{\mu\nu} \sqrt{|g|}\; \dx, \\
\int_M h \nmu \nnu \delta g^{\mu\nu} \sqrt{|g|} \; \dx & = \int_M \nmu
\nnu h \; \delta g^{\mu\nu} \sqrt{|g|} \; \dx.
\end{align}

To prove the first of these equations we use the Stokes theorem
and obtain
\begin{align}
\int_M h g_{\mu\nu} \Box \delta g^{\mu\nu} \sqrt{|g|} \; \dx &= \int_M h g_{\mu\nu} \nabla_{\alpha}\nabla^{\alpha} \delta g^{\mu\nu} \sqrt{|g|} \; \dx
=-\int_M \nabla_{\alpha}(h g_{\mu\nu})\nabla^{\alpha}\delta
g^{\mu\nu} \sqrt{|g|} \; \dx \nonumber \\
&= \int_M g_{\mu\nu}\nabla^{\alpha}\nabla_{\alpha}h \; \delta g^{\mu\nu} \sqrt{|g|}\; \dx = \int_M g_{\mu\nu} \Box h \; \delta g^{\mu\nu}\sqrt{|g|} \; \dx.
\end{align}
Here we have used $ \nabla_{\gamma}g_{\mu\nu}=0$ and $
\nabla^{\alpha}\nabla_{\alpha}=\nabla_{\alpha} \nabla^{\alpha}=
\Box$ to obtain the last integral.

To obtain the second equation we first introduce  vector 
\begin{equation}
N^\mu = h \nnu \delta g^{\mu\nu} - \nnu h \delta g^{\mu\nu}.
\end{equation}

From the above expression follows
\begin{align}
\nmu N^{\mu} & = \nmu (h \nnu \delta g^{\mu\nu} - \nnu h \delta g^{\mu\nu}) = \nmu h \nnu \delta g^{\mu\nu} + h \nmu \nnu \delta g^{\mu\nu} - \nmu \nnu h \; \delta g^{\mu\nu} - \nnu h \nmu \delta g^{\mu\nu} \nonumber\\
& = h \nmu \nnu \delta g^{\mu\nu} - \nmu \nnu h \; \delta
g^{\mu\nu}.
\end{align}

Integrating $\nmu N^{\mu}$ yields $\int_M \nmu N^{\mu} \sqrt{|g|} \; \dx = \int_{\partial M} N^{\mu}n_{\mu} \mathrm d\partial M $, where $n_\mu$ is the unit normal vector. Since $N^\mu|_{\partial M} = 0$ we have that the last
integral is zero, which completes the proof.
\end{proof}

\begin{lema}
Let $\theta$ and $\psi$ be scalar functions such that
$\delta \psi |_{\partial M} = 0$. Then one has
\begin{align} \label{theta:psi}
\int_M \theta \delta \Box \psi \sqrt{|g|} \; \dx & = \frac{1}{2} \int_M
g^{\alpha \beta} \partial_{\alpha} \theta \;
\partial_{\beta} \psi g_{\mu\nu} \delta g^{\mu\nu} \; \sqrt{|g|} \;
\dx  - \int_M \pmu \theta \; \pnu \psi \delta g^{\mu \nu} \sqrt{|g|} \;
\dx \nonumber \\
& + \int_M \Box \theta \; \delta \psi \sqrt{|g|} \; \dx + \frac{1}{2} \int_M g_{\mu\nu}\theta \Box \psi \delta g^{\mu\nu}
\sqrt{|g|}\; \dx.
\end{align}
\end{lema}

\begin{proof}
Since $\theta$ and $\psi$ are scalar functions such that
$\delta \psi |_{\partial M} = 0$ we have
\begin{align}
& \int_M \theta \delta \Box \psi \sqrt{|g|} \; \dx = \int_M \theta
\partial_{\alpha} \delta (\sqrt{|g|} g^{\alpha\beta} \partial_{\beta} \psi)\; \dx + \int_M \theta \delta \left(\frac{1}{\sqrt{|g|}} \right)
\partial_{\alpha} (\sqrt{|g|} g^{\alpha\beta} \partial_{\beta} \psi) \sqrt{|g|}\; \dx \nonumber\\
&= \int_M  \partial_{\alpha} (\theta \delta (\sqrt{|g|}
g^{\alpha\beta} \partial_{\beta} \psi))\; \dx - \int_M
\partial_{\alpha} \theta \; \delta (\sqrt{|g|} g^{\alpha\beta}
\partial_{\beta} \psi)\; \dx + \frac{1}{2} \int_M \theta g_{\mu\nu} \Box \psi \delta g^{\mu\nu}
\sqrt{|g|}\; \dx.
\end{align}
It is easy to see that $ \int_M  \partial_{\alpha} (\theta \delta
(\sqrt{|g|} g^{\alpha\beta} \partial_{\beta} \psi))\; \dx = 0$.
From this result it follows
\begin{align}
 \int_M \theta \delta \Box \psi \sqrt{|g|} \; \dx &= -\int_M g^{\alpha \beta} \partial_{\alpha} \theta \; \partial_{\beta} \psi \delta (\sqrt{|g|}) \; \dx
- \int_M \partial_{\alpha} \theta \;\partial_{\beta} \psi \delta g^{\alpha \beta} \sqrt{|g|} \; \dx \nonumber\\
&- \int_M g^{\alpha \beta} \sqrt{|g|} \partial_{\alpha} \theta \;
\partial_{\beta} \delta \psi \; \dx + \frac{1}{2} \int_M \theta g_{\mu\nu} \Box \psi \delta g^{\mu\nu} \sqrt{|g|}\; \dx \nonumber\\
&= \frac{1}{2} \int_M g^{\alpha \beta} \partial_{\alpha} \theta \;
\partial_{\beta} \psi g_{\mu\nu} \delta g^{\mu\nu} \; \sqrt{|g|} \; \dx
- \int_M \pmu \theta \; \pnu \psi \delta g^{\mu
\nu} \sqrt{|g|} \; \dx \nonumber\\
&- \int_M \partial_{\beta} (g^{\alpha \beta} \sqrt{|g|}
\partial_{\alpha} \theta \; \delta \psi) \; \dx  + \int_M \partial_{\beta} (g^{\alpha \beta} \sqrt{|g|} \partial_{\alpha} \theta) \;
 \delta \psi \; \dx \nonumber\\
&+ \frac{1}{2} \int_M g_{\mu\nu}\theta \Box \psi \delta g^{\mu\nu} \sqrt{|g|}\; \dx \nonumber\\
&= \frac{1}{2} \int_M g^{\alpha \beta} \partial_{\alpha} \theta \;
\partial_{\beta} \psi g_{\mu\nu} \delta g^{\mu\nu} \; \sqrt{|g|} \;
\dx - \int_M \pmu \theta \; \pnu \psi \delta
g^{\mu \nu} \sqrt{|g|} \; \dx \nonumber\\
&+ \int_M \Box \theta \; \delta \psi \sqrt{|g|} \; \dx + \frac{1}{2}
\int_M g_{\mu\nu}\theta \Box \psi \delta g^{\mu\nu} \sqrt{|g|}\; \dx.
\end{align}
At the end we have that
\begin{align}
&\int_M \theta \delta \Box \psi \sqrt{|g|} \; \dx = \frac{1}{2} \int_M
g^{\alpha \beta} \partial_{\alpha} \theta \;
\partial_{\beta} \psi g_{\mu\nu} \delta g^{\mu\nu} \; \sqrt{|g|} \;
\dx \nonumber \\
&- \int_M \pmu \theta \; \pnu \psi \delta g^{\mu \nu} \sqrt{|g|} \;
\dx + \int_M \Box \theta \; \delta \psi \sqrt{|g|} \; \dx + \frac{1}{2} \int_M g_{\mu\nu}\theta \Box \psi \delta g^{\mu\nu}
\sqrt{|g|}\; \dx.
\end{align}
\end{proof}

Now, after this preliminary work we can get the variation of
$S_1$.\\

\begin{lema} \label{vars1}
The variation of $S_1$ is
\begin{align*}
\delta S_1 &= -\frac{1}{2} \int_M g_{\mu\nu} \HH(R)\FF(\Box)\GG(R) \delta g^{\mu\nu} \sqrt{|g|} \; \dx + \int_M \left(R_{\mu\nu} \Phi - K_{\mu\nu} \Phi \right) \delta g^{\mu\nu} \sqrt{|g|} \; \dx \nonumber\\
&+ \frac{1}{2} \sum_{n=1}^{\infty} f_n \sum_{l=0}^{n-1} \int_M \Big(g_{\mu\nu} \left(\partial^\alpha \Box^l \HH(R) \partial_\alpha \Box^{n-1-l} \GG(R) + \Box^l \HH(R) \Box^{n-l} \GG(R)\right) \\
&-2 \pmu \Box^l \HH(R) \pnu \Box^{n-1-l} \GG(R)\Big) \delta g^{\mu\nu} \sqrt{|g|} \; \dx,
\end{align*}
where $K_{\mu\nu} = \nmu\nnu - g_{\mu\nu} \Box$, $ \, \, \Phi = \HH'(R)\FF(\Box)\GG(R) + \GG'(R)\FF(\Box)\HH(R)$ and  $'$ denotes derivative with respect to $R$.

\end{lema}

\begin{proof}
The variation of $S_1$ can be expressed as
\begin{align}
\delta S_1 &= \int_M \left(\HH(R) \FF(\Box)\GG(R) \delta (\sqrt{|g|}) + \delta (\HH(R)) \FF(\Box)\GG(R) \sqrt{|g|} +  \HH(R) \delta(\FF(\Box)\GG(R)) \sqrt{|g|}\right) \dx.
\end{align}

For the first two integrals in the last equation we have
\begin{align}
I_1 &= \int_M \HH(R) \FF(\Box)\GG(R) \delta (\sqrt{|g|}) \; \dx = -\frac{1}{2} \int_M g_{\mu\nu} \HH(R)\FF(\Box)\GG(R) \delta g^{\mu\nu} \sqrt{|g|} \; \dx, \\
I_2 &= \int_M \delta (\HH(R)) \FF(\Box)\GG(R) \sqrt{|g|} \; \dx = \int_M \HH'(R) \delta R \; \FF(\Box)\GG(R)\sqrt{|g|} \; \dx. \nonumber
\end{align}
Substituting $ h= \HH'(R) \; \FF(\Box)\GG(R)$ in  equation
\eqref{h:delta R} we obtain
\begin{align}
I_2 =& \int_M \Big( R_{\mu\nu} \HH'(R) \FF(\Box)\GG(R) - K_{\mu\nu}
\big(\HH'(R) \FF(\Box)\GG(R)\big)\Big) \delta g^{\mu\nu} \sqrt{|g|}
\; \dx.
\end{align}

The third integral can be presented using linear combination of the
following integrals
\begin{align}
J_n = \int_M \HH(R) \delta(\Box^n \GG(R)) \sqrt{|g|} \; \dx.
\end{align}

$J_0$ is the integral of the same form as $I_2$ so
\begin{align}
 J_0 = \int_M \Big( R_{\mu\nu} \GG'(R) \HH(R) - K_{\mu\nu} \big(\GG'(R) \HH(R)\big)\Big) \delta g^{\mu\nu} \sqrt{|g|} \;
 \dx.
\end{align}

For $n>0$, we can find $J_n$ using \eqref{theta:psi}.
In the first step we take $ \theta =\HH(R) $ and $ \psi=
\Box^{n-1} \GG(R)$ and obtain
\begin{align}
J_{n} &= \frac{1}{2} \int_M g^{\alpha \beta} \partial_{\alpha} \HH(R) \;
\partial_{\beta} \Box^{n-1} \GG(R) g_{\mu\nu} \delta g^{\mu\nu} \; \sqrt{|g|} \;
\dx - \int_M \pmu \HH(R) \; \pnu \Box^{n-1} \GG(R) \delta g^{\mu \nu}
\sqrt{|g|} \; \dx \nonumber  \\ &+ \int_M \Box \HH(R) \; \delta \Box^{n-1} \GG(R)
\sqrt{|g|} \; \dx + \frac{1}{2} \int_M g_{\mu\nu}\HH(R) \Box^{n} \GG(R) \delta g^{\mu\nu}
\sqrt{|g|}\; \dx.
\end{align}
In the second step we take $ \theta =\Box \HH(R) $ and $ \psi=
\Box^{n-2} \GG(R)$ and get the third integral in this formula,
etc. Using \eqref{theta:psi} $n$ times one obtains
\begin{align}
J_n &= \frac{1}{2} \sum_{l=0}^{n-1} \int_M \left( g_{\mu\nu} \partial^\alpha \Box^l \HH(R) \partial_\alpha \Box^{n-1-l} \GG(R) +g_{\mu\nu} \Box^l \HH(R) \Box^{n-l} \GG(R) -2 \pmu \Box^l \HH(R) \pnu \Box^{n-1-l} \GG(R) \right)\delta g^{\mu\nu} \sqrt{|g|} \; \dx \nonumber \\
&+ \int_M \Big( R_{\mu\nu} \GG'(R) \Box^n \HH(R) - K_{\mu\nu}
\big(\GG'(R) \Box^n \HH(R)\big)\Big) \delta g^{\mu\nu} \sqrt{|g|}
\; \dx.
\end{align}
Using the equation \eqref{h:delta R} we obtain the last integral in the above formula.
Finally, one can put everything together and obtain
\begin{align}
\delta S_1 &= I_1 + I_2 + \sum_{n=0}^{\infty} f_n J_n = -\frac{1}{2} \int_M g_{\mu\nu} \HH(R)\FF(\Box)\GG(R) \delta g^{\mu\nu} \sqrt{|g|} \; \dx + \int_M \left(R_{\mu\nu} \Phi - K_{\mu\nu} \Phi \right) \delta g^{\mu\nu} \sqrt{|g|} \; \dx \nonumber\\
&+ \frac{1}{2} \sum_{n=1}^{\infty} f_n \sum_{l=0}^{n-1} \int_M \Big(g_{\mu\nu} \left(\partial^\alpha \Box^l \HH(R) \partial_\alpha \Box^{n-1-l} \GG(R) + \Box^l \HH(R) \Box^{n-l} \GG(R)\right) \nonumber\\
&-2 \pmu \Box^l \HH(R) \pnu \Box^{n-1-l} \GG(R)\Big) \delta g^{\mu\nu} \sqrt{|g|} \; \dx. &\qedhere
\end{align}
\end{proof}


\begin{theo} \label{t1}
The variation of the functional \eqref{lag:1} is equal to zero iff
\begin{align} \label{EOM}
&\frac{G_{\mu\nu} +\Lambda g_{\mu\nu}}{16 \pi G} -\frac{1}{2} g_{\mu\nu} \HH(R)\FF(\Box)\GG(R) + \left(R_{\mu\nu} \Phi - K_{\mu\nu} \Phi \right) \nonumber\\
&+ \frac{1}{2} \sum_{n=1}^{\infty} f_n \sum_{l=0}^{n-1} \big( g_{\mu\nu}  \partial^\alpha \Box^l \HH(R) \partial_\alpha \Box^{n-1-l} \GG(R) - 2 \pmu \Box^l \HH(R) \pnu \Box^{n-1-l} \GG(R) + g_{\mu\nu}
\Box^l \HH(R) \Box^{n-l} \GG(R) \big) = 0.
\end{align}
\end{theo}

\begin{proof}
Since we have $\delta S = \frac 1{16\pi G} \delta S_0 + \delta S_1$ the theorem follows from Lemmas \ref{lema:3} and \ref{vars1}.
\end{proof}

\section{Signature $(1,3)$}

In the physics settings, where functional $S$ represents an action, theorem \ref{t1} gives the equations of motion. From this point we assume that manifold $M$ is the four-dimensional homogeneous and isotropic one with signature $(1,3)$. Then the metric  has the Friedmann-Lema\^{\i}tre-Robertson-Walker (FLRW) form:
\begin{equation}
  \mathrm ds^2 = - \mathrm dt^2 + a(t)^2 \left( \frac{\mathrm dr^2}{1-kr^2} + r^2 \mathrm d\theta^2 + r^2 \sin^2 \theta \mathrm d\varphi^2\right).
\end{equation}

\begin{theo}
Suppose that manifold $M$ has the $FLRW$ metric. Then the expression \eqref{EOM} has two linearly independent equations:
\begin{align} \label{trace}
&\frac{ 4 \Lambda-R}{16 \pi G} -2 \HH(R)\FF(\Box)\GG(R) + \left(R \Phi + 3\Box \Phi \right) \nonumber \\
\phantom{A}&+ \sum_{n=1}^{\infty} f_n \sum_{l=0}^{n-1} \big( \partial^{\mu} \Box^l \HH(R) \pmu \Box^{n-1-l} \GG(R) + 2 \Box^l \HH(R) \Box^{n-l} \GG(R) \big) = 0,
\end{align}

\begin{align} \label{EOM-00}
&\frac{G_{00}+ \Lambda g_{00}}{16 \pi G}  -\frac{1}{2} g_{00} \HH(R)\FF(\Box)\GG(R) + \left(R_{00} \Phi - K_{00} \Phi \right) \nonumber \\
&+ \frac{1}{2} \sum_{n=1}^{\infty} f_n \sum_{l=0}^{n-1} \big( g_{00} \partial^\alpha \Box^l \HH(R) \partial_\alpha \Box^{n-1-l} \GG(R) - 2 \partial_0 \Box^l \HH(R) \partial_0 \Box^{n-1-l} \GG(R) +
g_{00} \Box^l \HH(R) \Box^{n-l} \GG(R) \big) = 0.
\end{align}
\end{theo}

\begin{proof}
The $FLRW$ metric satisfies $R_{\mu\nu} = \frac R4 g_{\mu\nu}$ and scalar curvature $R = 6 \left( \frac{\ddot a}a + \left(\frac{\dot a}a\right)^2 + \frac k{a^2}\right)$ depends only on $t$, hence equations \eqref{EOM} for $\mu \neq \nu$ are trivially satisfied. On the other hand, equations with indices $11$, $22$ and $33$ can be rewritten as

\begin{align*}
&g_{\mu\mu} \Big(\frac{- \frac R4 +\Lambda}{8 \pi G} - \HH(R)\FF(\Box)\GG(R) + \frac R2  \Phi +  \sum_{n=1}^{\infty} f_n \sum_{l=0}^{n-1} \big( \partial^\alpha \Box^l \HH(R) \partial_\alpha \Box^{n-1-l} \GG(R) + \Box^l \HH(R) \Box^{n-l} \GG(R) \big)\Big) = 0.
\end{align*}
Therefore these three equations are linearly dependent and there are only two linearly independent equations. The most convenient choice is the trace and $00$-equation.
\end{proof}

\begin{corollary}
For $\HH(R)= R^p $ and $\GG(R)= R^q$ the action \eqref{lag:1} becomes
\begin{equation} \label{lag:2}
S = \displaystyle \int_M \Big(\frac{R-2 \Lambda}{16 \pi G} +  R^p \mathcal{F}(\Box)R^q \Big) \sqrt{|g|} \;\dx,
\end{equation}
and equations of motion are
\begin{equation} \begin{aligned} \label{EOM:2}
&\frac{1}{16 \pi G}(G_{\mu\nu} + \Lambda g_{\mu\nu}) -\frac{1}{2} g_{\mu\nu} R^p \FF(\Box)R^q + \left(R_{\mu\nu} \Phi - K_{\mu\nu} \Phi \right) \\
&+ \frac{1}{2} \sum_{n=1}^{\infty} f_n \sum_{l=0}^{n-1} \big( g_{\mu\nu} \partial^\alpha \Box^l R^p \partial_\alpha \Box^{n-1-l} R^q - 2 \pmu \Box^l R^p \pnu \Box^{n-1-l} R^q + g_{\mu\nu} \Box^l R^p
\Box^{n-l} R^q \big) = 0,
\end{aligned} \end{equation}
where  $\Phi = pR^{p-1}\FF(\Box)R^{q} + qR^{q-1}\FF(\Box)R^{p}$.
\end{corollary}

\begin{corollary}
For $\HH(R)= R^p $ and $\GG(R)= R^q$ the equations of motion \eqref{EOM:2} are equivalent to the following two equations:
\begin{equation} \begin{aligned} \label{trace:2}
&\frac{1}{16 \pi G} ( 4 \Lambda - R) -2 R^p \FF(\Box)R^q + \left(R \Phi + 3\Box \Phi \right) + \sum_{n=1}^{\infty} f_n \sum_{l=0}^{n-1} \big( \partial^\mu \Box^l R^p \pmu \Box^{n-1-l} R^q + 2 \Box^l R^p \Box^{n-l} R^q \big) = 0,
\end{aligned} \end{equation}

\begin{equation} \begin{aligned} \label{EOM-00:2}
&\frac{1}{16 \pi G} ( G_{00} + \Lambda g_{00}) -\frac{1}{2} g_{00} R^p \FF(\Box)R^q  + \left(R_{00} \Phi - K_{00} \Phi \right) \\
&+ \frac{1}{2} \sum_{n=1}^{\infty} f_n \sum_{l=0}^{n-1} \big( g_{00} \partial^\alpha \Box^l R^p \partial_\alpha \Box^{n-1-l} R^q + g_{00} \Box^l R^p \Box^{n-l} R^q  - 2 \partial_0 \Box^l R^p \partial_0 \Box^{n-1-l} R^q\big) = 0.
\end{aligned} \end{equation}
\end{corollary}

\section{Constant scalar curvature}

\begin{theo} Let $R = R_0 = constant$. Then, solution of equations of motion \eqref{EOM:2} has the  form
\begin{enumerate}
\item  For $R_0 > 0$,
$a(t) = \sqrt{\frac{6k}{R_0} + \sigma e^{\sqrt {\frac {R_0}3} t} + \tau e^{- \sqrt {\frac {R_0}3} t}}$,
where $  9k^2 = R_0^2\sigma\tau $, $\;\sigma, \tau \in \Real$.
\item For $R_0 = 0$,
$a(t) = \sqrt{ -k t^2 + \sigma t + \tau}$, where $ \sigma^2 + 4k \tau=0 $, $\;\sigma, \tau \in \Real$
\item For $R_0 < 0$,
$a(t) = \sqrt{\frac{6k}{R_0} + \sigma \cos {\sqrt {\frac{-R_0}3} t} + \tau \sin {\sqrt {\frac {-R_0}3} t}}$, where $ 36k^2 = R_0^2(\sigma^2+ \tau^2)$, $\;\sigma, \tau \in \Real$,
\end{enumerate}
where $k$ is curvature parameter.

\end{theo}
\begin{proof}
Since $R = R_0$ one has
\begin{equation}
6\Big(\frac{\ddot a}{a} + \big(\frac{\dot a}{a}\big)^2 +
\frac{k}{a^2}\Big)= R_0.
\end{equation}

The change of variable $b(t) = a^2(t)$ yields second order linear
differential equation with constant coefficients
\begin{equation}
3 \ddot b - R_0 b = -6k.
\end{equation}

Depending on the sign of $R_0$ we have the following solutions for
$b(t)$
\begin{equation} \begin{aligned} \label{11.03.13:1}
R_0 > 0,  \qquad & b(t) = \frac{6k}{R_0} + \sigma e^{\sqrt {\frac {R_0}3} t} + \tau e^{- \sqrt {\frac {R_0}3} t}, \\
R_0 = 0,  \qquad & b(t) = -k t^2 + \sigma t + \tau, \\
R_0 < 0,  \qquad & b(t) = \frac{6k}{R_0} + \sigma \cos {\sqrt {\frac{-R_0}3} t} + \tau \sin {\sqrt {\frac {-R_0}3} t}.
\end{aligned}\end{equation}

Putting $R = R_0 = const$ into  \eqref{trace:2} and
\eqref{EOM-00:2} one obtains the following two equations:
\begin{equation} \begin{aligned} \label{11.03.13:2}
 f_0 R_0^{p+q} \Big(p+q -2\Big) &= \frac{R_0-4\Lambda}{16 \pi G}, \qquad
 f_0 R_0^{p+q-1} \Big( \frac 12 R_0 + (p+q)R_{00} \Big) &=
\frac{-G_{00} + \Lambda}{16 \pi G} .
\end{aligned} \end{equation}

Equations \eqref{11.03.13:2} will have a solution if and only if
\begin{equation} \label{iff}
R_0^{p+q-1} (R_0+ 4R_{00})(R_0 + (2\Lambda-R_0) (p+q)) = 0.
\end{equation}

In the first case we take $R_0+ 4R_{00} = 0$ that yields the
following conditions on the parameters $\sigma$ and $\tau$:

\begin{equation} \begin{aligned} \label{11.03.13:3}
R_0 > 0,  \qquad & 9k^2 = R_0^2\sigma\tau, \\
R_0 = 0,  \qquad & \sigma^2 + 4k \tau=0, \\
R_0 < 0,  \qquad & 36k^2 = R_0^2(\sigma^2+ \tau^2). &\qedhere
\end{aligned}\end{equation}
\end{proof}

Solutions given in \eqref{11.03.13:1} together with conditions
\eqref{11.03.13:3} restrict the possibilities for the parameter
$k$.
\begin{theo}
\begin{enumerate}
\item If $ R_0 > 0 $ then for $k=0$ there is solution with constant Hubble parameter,
for $k=+1$ the solution is
$a(t) = \sqrt{\frac{12}{R_0}} \cosh \frac 12 \left(\sqrt{\frac{R_0}{3}}
t + \varphi\right)$ and for $k=-1$ it is
$ \, a(t) = \sqrt{\frac{12}{R_0}} \left|\sinh \frac 12 \left(\sqrt{\frac{R_0}{3}} t + \varphi\right)\right|$, where $\sigma+\tau = \frac{6}{R_0} \cosh \varphi$ and $\sigma-\tau =
\frac{6}{R_0} \sinh\varphi$.
\item If $ R_0 = 0 $ then for $k=0$ the solution is $a(t) = \sqrt \tau = const$ and
for $k=-1$ the solution is $a(t)= |t+ \frac \sigma2|$.
\item If $ R_0 < 0 $ then for $k=-1$ the solution is $a(t) = \sqrt{\frac{-12}{R_0}} \left|\cos \frac 12
\left(\sqrt{-\frac{R_0}{3}} t - \varphi\right)\right|$, where $\sigma = \frac{-6}{R_0} \cos \varphi$ and $\tau = \frac{-6}{R_0}
\sin \varphi$.
\end{enumerate}
\end{theo}

\begin{proof}

Let $R_0>0$. Set $k=0$ then we obtain solution with constant Hubble parameter.
Alternatively, if we set $k=+1$ then there is $\varphi$ such that
$\sigma+\tau = \frac{6}{R_0} \cosh \varphi$ and $\sigma-\tau =
\frac{6}{R_0} \sinh\varphi$. Moreover, we obtain
\begin{equation} \begin{aligned}
b(t) &= \frac{12}{R_0} \cosh^2 \frac 12 \left(\sqrt{\frac{R_0}{3}} t + \varphi\right), \qquad
a(t) &= \sqrt{\frac{12}{R_0}} \cosh \frac 12 \left(\sqrt{\frac{R_0}{3}} t + \varphi\right).
\end{aligned}\end{equation}
At the end, if we set $k=-1$ one can transform $b(t)$ to
\begin{equation} \begin{aligned}
b(t) &= \frac{12}{R_0} \sinh^2 \frac 12 \left(\sqrt{\frac{R_0}{3}} t + \varphi\right), \qquad
a(t) &= \sqrt{\frac{12}{R_0}} \left|\sinh \frac 12 \left(\sqrt{\frac{R_0}{3}} t + \varphi\right)\right|.
\end{aligned}\end{equation}

Let $R_0=0$. If  $k=0$ then function $b(t)$ and consequently
$a(t)$ become constants which leads to a solution $a(t) = \sqrt
\tau = const$. On the other hand if $k \neq 0$ then we can write
\begin{equation} \begin{aligned}
b(t) &= -k(t - \frac{\sigma}{2k})^2.
\end{aligned}\end{equation}
If $k=+1$ then there are no solutions for the scale factor $a(t) ,$ because $b(t)
\leq 0$. On the other hand, when $k=-1$ the scalar factor becomes
\begin{equation} \begin{aligned}
a(t) &= |t+ \frac \sigma2|.
\end{aligned}\end{equation}

In the last case, let $R_0 < 0$. If $k=-1$ we can find $\varphi$ such that
$\sigma = \frac{-6}{R_0} \cos \varphi$ and $\tau = \frac{-6}{R_0}
\sin \varphi$ and rewrite $a(t)$ and $b(t)$ as
\begin{equation} \begin{aligned}
b(t) &= \frac{-12}{R_0} \cos^2 \frac 12 \left(\sqrt{-\frac{R_0}{3}} t - \varphi\right), \qquad
a(t) &= \sqrt{\frac{-12}{R_0}} \left|\cos \frac 12 \left(\sqrt{-\frac{R_0}{3}} t - \varphi\right)\right|.
\end{aligned}\end{equation}
In the case $k=+1$ one can transform $b(t)$ to $ b(t) = \frac{12}{R_0} \sin^2 \frac 12 \left(\sqrt{-\frac{R_0}{3}} t -
\varphi\right)$, which is non positive and hence yields no solutions.
\end{proof}

\begin{theo} \label{theorem 4.3}
If in \eqref{iff} we take
\begin{equation} \begin{aligned} \label{11.03.13:4}
R_0^{p+q-1}\big(R_0 + (p+q) (2\Lambda-R_0)\big) = 0
\end{aligned}\end{equation} then: \\
\begin{enumerate}
\item For $p+q\geq 1$ there is obvious solution $R_0=0$. In particular if $p+q=1$ then \eqref{11.03.13:4} is satisfied for any
$R_0 \neq 0$ if  $\Lambda = 0$.
\item For $p+q=0$ there is no solution.
\item For $p+q \neq 0,1$ there is a unique value $ R_0 = \displaystyle \frac{2\Lambda(p+q)}{p+q-1} $ that  gives a solution. Since $p$ and $q$ are integers the
value of $R_0$  in the last equation is always positive,  and for $ k= 0$ the
solution $b(t)$ is a linear combination of exponential functions.
\end{enumerate}
\end{theo}

Proof of this theorem is evident.




\end{document}